\documentclass[ french, 10pt]{article}
\usepackage{amsmath}
\usepackage{amssymb}
\usepackage{amsthm}
\usepackage{graphicx}
\usepackage{float}
\usepackage{url}

\newtheorem{theorem}{Theorem}

\newtheorem{remark}{Remark}

\newtheorem{conjecture}{Conjecture}

\newtheorem{lemma}{Lemma}

\begin{document}

\begin{center}

{\Large Stochastic Loewner Evolutions, Fuchsian Systems and Orthogonal Polynomials}

\vspace{5mm}

{\large Igor Loutsenko and Oksana Yermolayeva}

\vspace{5mm}

\end{center}

\begin{abstract}
We  find a wide class of Levy-Loewner evolutions for which the value of integral means beta-spectrum $\beta(q)$ at $q=2$ is the maximal real eigenvalue of a three-diagonal matrix. The second moments of derivatives of corresponding conformal mappings are expressed through solutions of matrix Fuchsian systems with three singular points.
\end{abstract}

\begin{section}{Introduction}

\label{Introduction}

The Schramm-Loewner Evolution is so far the only model where a multi-fractal spectrum is not trivial and known explicitly (see e.g. \cite{BS, BDZ, D, DNNZ, Has, LSLE, LYSLE}). Explicit computation of harmonic spectrum is almost impossible task for deterministic fractals and it is a difficult task even for random fractals. Apart from the case of the Schramm-Loewner evolution, two exact results for the second moment of harmonic measure were obtained in the study of Levy-Loewner Evolution (LLE) \cite{DNNZ, LYLLE}. Extending this study, we will present an infinite set of non-trivial examples.

The Bounded Whole-Plane Levy-Loewner evolution (also called exterior whole-plane LLE) \footnote{For a short introduction to the whole-plane LLE see e.g. \cite{LYLLE}, \cite{JS}. A good introduction to chordal LLE can be found in \cite{ORGK, ROKG}. For a quick introduction to Levy processes see e.g. \cite{A}.} is a stochastic process of the growth of the curve out of the origin $z=0$ in the complex plane $z$ (see Figure \ref{LLE}).

The curve is described by time dependent conformal mapping $z=F(w,t)$ from exterior of the unit circle $|w|>1$ in the $w$-plane to the complement of the curve in the $z$ plane. This mappimg obeys the Levy-Loewner equation
\begin{equation}
\frac{\partial F(w, t)}{\partial t}=w\frac{\partial F(w,
t)}{\partial w}\frac{w+e^{i L(t)}}{w-e^{i L(t)}}, \quad \lim_{w\to\infty} F'(w,t)=e^t, \quad -\infty < t < \infty,
\label{LE}
\end{equation}
where $L(t)$ is a Levy process and prime denotes the $w$-derivative. The limit $t \to -\infty$ corresponds to the origin of the curve.

Here we consider the Levy processes
without a drift. Without loss of generality we set
\begin{equation}
\langle L(t) \rangle = 0,
\label{Drift=0}
\end{equation}
where $\langle \rangle$ denote expectation (ensemble average).

The LLE is a conformally invariant stochastic process in the sense that the time evolution is consistent with composition of conformal maps. When $L(t)$ is a continuous function of time, the conformal
mapping $z=F(w,t)$ describes growth of a random continuous curve. The only continuous (modulo uniform drift) process of Levy type is
the Brownian motion
\begin{equation}
L(t)=\sqrt{\kappa} B(t), \quad
\langle (B(t+\tau)-B(t))^2 \rangle = |\tau|,
\label{Bt}
\end{equation}
where the positiive parameter $\kappa$ is the ``temperature" of the Brownian motion. The stochastic Loewner evolution driven by Brownian motion is called
Schramm-Loewner Evolution or SLE$_\kappa$. Since it describes
non-branching planar stochastic curves with a conformally-invariant
probability distribution, SLE is a useful tool for description of
boundaries of critical clusters in two-dimensional equilibrium
statistical mechanics. In this picture, different $\kappa$
correspond to different classes of models of statistical mechanics
(a good introduction to SLE for physicists can be found e.g. in \cite{C,G} as
well as mathematical reviews are given e.g. in \cite{GL,GL2}).

On the other hand, in the general case, when $L(t)$ is a discontinuous function of time the conformal mapping describes stochastic (infinitely) branching curve.

\begin{figure}
\centering
\includegraphics[width=115mm]{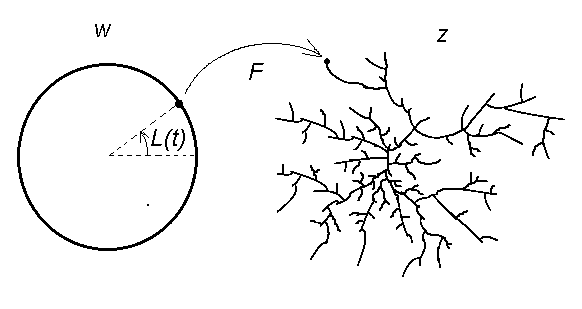}
 \caption{\small Bounded whole-plane LLE: Time dependent mapping $F$ from the exterior of the unit circle in the $w$-plane to the complement of the curve in $z$-plane. Point $w=e^{iL(t)}$ on the unit circle is mapped to the growing tip of the curve.}
 \label{LLE}
 \end{figure}

The unbounded (or interior) version of the whole-plane LLE is an
inversion $F(w,t)\to 1/F(1/w,t)$ of the above bounded version. Obviously, this mapping from an interior of the unit disc in the $w$-plane to the complement of stochastic curve which grows from infinity towards the origin in the $z$-plane also satisfies the Loewner equation (\ref{LE}) with different asymptotic conditions given now at $w=0$
$$
F'(w=0,t)=e^{-t}.
$$
This version of LLE has been studied mainly due to its relationship with the problem of Bieberbach
coefficients of conformal mappings \cite{DNNZ,LSLE,LYLLE}.

To get the multi-fractal spectrum of LLE one has to find first the so-called ``$\beta$-spectrum" \footnote{For introduction to multi-fractal analysis and relationship between different kinds of multi-fractal spectra see e.g. \cite{BS}, \cite{HHD}, \cite{HJKPS}, \cite{Has}, \cite{Wiki}.}: The integral means $\beta(q)$-spectrum of the
domain is defined through the $q$th moment of a derivative of conformal mapping at the unit
circle (i.e. at $|w|\to 1$) as follows
\begin{equation}
\beta(q)=\overline{\lim}_{\epsilon\to 0\pm}\frac{\log
\int_0^{2\pi}\langle\left|F'\left(e^{\epsilon+i\varphi}\right)\right|^q\rangle
d\varphi}{-\log \pm\epsilon},
\label{beta}
\end{equation}
where $\pm$ signs correspond to bounded and unbounded whole-plane LLE (in what follows the upper/lower signs will correspond to the bounded/unbounded version respectively).

To find $\beta(q)$-spectrum (\ref{beta}) one needs to estimate moments of derivative: One can show that the value
$$
\rho=e^{\mp qt}\langle\left|F'\left(e^{iL(t)}w,t\right)\right|^q\rangle,
$$
is time independent and is a function of \footnote{$\bar w$ denotes complex conjugate of $w$} $w,\bar w$ and $q$ only \cite{BS, BDZ, DNNZ, LSLE, LYSLE, LYLLE}. Moreover it satisfies the linear integro-differential equation
\begin{equation}
\mathcal{L}^\pm\rho=0, 
\label{Lrho}
\end{equation}
where the linear operator $\mathcal{L}^\pm$ equals
\begin{equation}
\mathcal{L}^\pm=-\hat\eta+w\frac{w+1}{w-1}\partial_w+\bar w\frac{\bar w+1}{\bar
w-1}\partial_{\bar w}-\frac{q}{(w-1)^2}-\frac{q}{(\bar w-1)^2}+q\mp q .
\label{Leta}
\end{equation}
In (\ref{Leta}), the linear operator $\hat\eta$ acts on functions of $w, \bar w$ as follows (for details see e.g. \cite{BS}, \cite{DNNZ} or \cite{LYLLE})
$$
\hat\eta\rho(w, \bar w)=\lim_{t\to
0}\frac{1}{t}\int_0^{2\pi}\left(\rho(w, \bar
w)-\rho\left(e^{i\varphi}w, e^{-i\varphi}\bar
w\right)\right) P(\varphi, t) d\varphi,
$$
where $P(\varphi,t)$ is the probability density that $L(t)=\varphi$ under
condition $L(0)=0$. One can present operator $\hat \eta$ in the more convenient (integro-differential) form as
\begin{equation}
\hat\eta\rho(w, \bar w)=\frac{\kappa}{2}\left(w\partial_w-\bar w\partial_{\bar
w}\right)^2\rho(w,\bar w)+\int_{-\pi}^{\pi}\left(\rho(w, \bar
w)-\rho\left(e^{i\varphi}w, e^{-i\varphi}\bar
w\right)\right) d\eta(\varphi), \quad d\eta(\varphi)\ge 0
\label{eta}
\end{equation}
where $d\eta(\varphi)$ is a symmetric Levy measure. In polar coordinates $(r,\phi)$, such that $w=re^{i\phi}, \bar{w} =re^{-i\phi}$, the operator $\hat\eta$ acts only wrt angular variable $\phi$, i.e. $\hat \eta$ commutes with $r$ and
$$
\hat\eta f(\phi)=-\frac{\kappa}{2}\partial_\phi^2f(\phi)+\int_{-\pi}^{\pi}\left(f(\phi)-f(\phi+\varphi)\right) d\eta(\varphi)
$$
whatever is dependence of $f$ on $r$. In other words, it acts diagonally on the basis $w^n\bar w^m$
\begin{equation}
\hat \eta w^n\bar w^m = \eta_{n-m} w^n\bar w^m, \quad n,m \in \mathbb{Z} .
\label{hat_eta_w_w}
\end{equation}
Since we consider only Levy processes without drift (\ref{Drift=0}) (i.e. with Levy measure symmetric wrt reflection $\phi\to-\phi$), the characteristic coefficients $\eta_n$ are real, non-negative and symmetric
$$
\eta_n=\bar \eta_n=\eta_{-n}, \quad \eta_0=0.
$$
They express through the Levy measure as
\begin{equation}
\eta_n=\frac{\kappa n^2}{2}+\int_{-\pi}^{\pi}(1-\cos n\phi)d\eta(\phi) .
\label{LK}
\end{equation}
The above equation is a particular case of Levy-Khintchine formula for a class of Levy processes we deal with.

The probability density $P(\phi,t)=\frac{1}{2\pi}\sum_{n=-\infty}^\infty e^{-\eta_n t+in\phi}$ is a fundamental solution of the integro-differential equation of the parabolic type on the unit circle
$$
\partial_t P(\phi,t)=-\hat\eta P(\phi,t), \quad P(\phi,0)=\delta(\phi) .
$$

Note that the particular case $d\eta(\phi)=0$ corresponds to the Schramm-Loewner evolution SLE$_\kappa$. Here, eq. (\ref{Lrho},\ref{Leta}) becomes the second-order differential equation with
$$
\hat\eta=-\frac{\kappa}{2}\partial_\phi^2, \quad \eta_n=\frac{\kappa n^2}{2},
$$
which allows to find exact form of the multi-fractal spectrum of the SLE$_{\kappa}$ \cite{BS,BDZ,Has,LYSLE}.

In the case of the bounded whole-plane LLE one can use analyticity of $\rho(w,\bar w)$ at infinity
as boundary conditions for linear equation (\ref{Lrho},\ref{Leta}). Namely (see e.g. \cite{LYLLE} for details), in this version of LLE $\rho$ has following asymptotic expansion at $w\to\infty$
$$
\rho=\sum_{i=0, j=0}^\infty \frac{\rho_{i,j}}{w^i\bar w^j}, \quad \rho_{0,0}=1
$$
with expansion coefficients $\rho_{i,j}$ fixed uniquely by equation (\ref{Lrho},\ref{Leta}). In other word, once analytic, non-vanishing at infinity solution of (\ref{Lrho},\ref{Leta}) is found, it gives (modulo constant factor) expectation of the moment of the derivative of conformal mappings of the bounded whole-plane LLE.

Similarly, for the unbounded version of LLE one has to look for non-vanishing analytic solution of the corresponding equation, this time at the origin $w=0$
$$
\rho=\sum_{i=0, j=0}^\infty \rho_{i,j}w^{i}\bar w^{j}, \quad \rho_{0,0}=1 .
$$
Similarly to the bounded case, such a solution is unique.

\end{section}

\begin{section}{Beta Spectrum of the Unbounded Whole-Plane LLE and Fuchsian Systems}
\label{q2}

Let us try to find the value of $\beta$-spectrum at $q=2$: Representing $\rho(w, \bar w)$ in the form
\begin{equation}
\rho(w, \bar w)=(1-w)(1-\bar w)\Theta(w, \bar w),
\label{rho1}
\end{equation}
from (\ref{Lrho},\ref{Leta}) with $q=2$ we get
\begin{equation}
-\hat\eta[(1-w)(1-\bar w)\Theta]+(w+1)(\bar w
-1)w\frac{\partial\Theta}{\partial w}+(\bar w+1)(w -1)\bar
w\frac{\partial\Theta}{\partial \bar w}+3(2w \bar w-w-\bar
w)\Theta=0 ,
\label{Ltheta1}
\end{equation}
where $\Theta$ is the series
\begin{equation}
\Theta=\theta_0(\xi)+\sum_{i=1}^{\infty}\left(w^i+\xi^i w^{-i}
\right) \theta_i(\xi), \quad \xi=r^2=w\bar w .
\label{theta_1}
\end{equation}
Substituting it into (\ref{Ltheta1}), and  taking into account the facts that operator $\hat \eta$ commutes with $\xi=r^2=w\bar w$ and that, according to (\ref{hat_eta_w_w}), $\hat \eta
[w^i]=\eta_i w^i$, from (\ref{Ltheta1},\ref{theta_1}) we get a three-term
differential recurrence relation for $\theta_i(\xi)$.
\begin{equation}
2\xi(\xi-1)\theta'_i(\xi)-\left(\eta_i+i+(\eta_i-i-6)\xi\right)\theta_i(\xi)+\xi(\eta_i+i-2)\theta_{i+1}(\xi)+(\eta_i-i-2)\theta_{i-1}(\xi)=0, \quad \theta_{-i}(\xi)=\xi^i\theta_i(\xi)
\label{rec_diff}
\end{equation}
It is easy to see that this recurrence relation has a solution that truncates at $i=N$, i.e. $\theta_i=0$ for $i \ge N$, when $\eta_N=N+2$ for some $N$.

The simplest truncation happens when $N=1$, i.e. when $\eta_1=3$, which corresponds to the case of conjecture by authors of \cite{DNNZ} proved by us in \cite{LYLLE}. Before studying the general case of an arbitrary $N$ we will consider separately the case of truncation at $N=2$. Here not only $\beta(2)$, but also $\rho(w,\bar w)$ can be found explicitly for arbitrary $\eta_i$, $i\not=2$. We have the following
\begin{theorem}
\label{N=2}
For the unbounded whole plane LLE with $\eta_2=4$ the integral means $\beta$-spectrum has the following value at $q=2$
\begin{equation}
\beta(2)=\frac{6-\eta_1+\sqrt{\eta_1^2-4\eta_1+12}}{2} ,
\label{betaN2}
\end{equation}
\end{theorem}

\begin{proof}
The proof is based on explicit computation of $\rho(w,\bar w)$. It is convenient to re-parametrize $\eta_1\ge 0$ through $\beta$ given by lhs of eq.(\ref{betaN2}), i.e.
\begin{equation}
\eta_1=\frac{\beta^2-6\beta+6}{2-\beta}, \quad 2<\beta< 3+\sqrt{3} ,
\label{eta1beta}
\end{equation}
and introduce new dependent variables $f_0, f_1$ such that $\theta_i(\xi)=(1-\xi)^{-\beta} f_i(\xi)$.

In the case of truncation at $N=2$ (i.e. when $\eta_2=4$), system (\ref{rec_diff}) writes as
\begin{equation}
\xi(\xi-1)f'_0+(3-\beta)\xi f_0-2\xi f_1=0
\label{dtheta0}
\end{equation}
\begin{equation}
2\xi(\xi-1) f'_1-\left(\eta_1+1+(\eta_1-7+2\beta)\xi\right) f_1+(\eta_1-3) f_0=0
\label{dtheta1}
\end{equation}
Then expressing $f_1$ through $f_0, f'_0$ into (\ref{dtheta0}) and substituting the result in (\ref{dtheta1}) we get the hypergeometric equation for $f_0$
$$
\xi(\xi-1)f''_0+\left((a+b+1)\xi-c\right)f'_0+abf_0=0
$$
where
\begin{equation}
a=\frac{\beta^2-5\beta+8}{2(2-\beta)}, \quad b=3-\beta, \quad c=\frac{\beta^2-7\beta+8}{2(2-\beta)}
\label{abc}
\end{equation}
Since we are looking for $\rho(w,\bar w)$ which is non-vanishing and analytic at $w=0$, from all linearly independent solutions of the above hypergeometric equation we choose the one which is non-vanishing and analytic at $\xi=0$, i.e.
$$
f_0={}_2{\rm F}_1(a,b,c;\xi)
$$
Then from (\ref{dtheta0}) it follows that
$$
f_1=\frac{\xi-1}{2}f'_0+\frac{3-\beta}{2}f_0=\frac{ab}{2c}(\xi-1){}_2{\rm F}_1(a+1,b+1,c+1;\xi)+\frac{3-\beta}{2}{}_2{\rm F}_1(a,b,c;\xi)
$$
and according to (\ref{rho1},\ref{theta_1})
$$
\rho(w, \bar w)=(1-w)(1-\bar w)(1-w\bar w)^{-\beta}\left(f_0(w\bar w)+(w+\bar w)f_1(w \bar w)\right) .
$$
According to (\ref{beta},\ref{rho1},\ref{theta_1}), the value of the $\beta$-spectrum equals the blowup rate (also called ``polynomial" growth rate) of $\theta_0-\theta_1=(1-\xi)^{-\beta} (f_0(\xi)-f_1(\xi))$ at $\xi\to 1$. Using the Gauss identity
$$
{}_2{\rm F}_1(a,b,c;1)=\frac{\Gamma(c)\Gamma(c-a-b)}{\Gamma(c-a)\Gamma(c-b)},\quad \Re (c-a-b)>0 ,
$$
and taking eqs. (\ref{eta1beta}, \ref{abc}) into account, it is easy to see that $f_0-f_1$ is finite and non-vanishing at $\xi=1$. Therefore, the blowup rate of $\theta_0-\theta_1$ equals $\beta$ and is given by lhs of eq.(\ref{betaN2}). This completes the proof.

\end{proof}

To consider he case of truncation at $i=N$ (i.e. when $\eta_N=N+2$), we note that in such a case the system of ODEs (\ref{rec_diff}) is a Fuchsian system with three singular points $\xi=0,1,\infty$. This system rewrites in the standard form as
\begin{equation}
\theta'(\xi)=\frac{A\theta(\xi)}{\xi}-\frac{B\theta(\xi)}{\xi-1},
\label{Fuchs}
\end{equation}
where $\theta$ is $N$-dimensional vector $\theta(\xi)=(\theta_0(\xi),\theta_1(\xi), \dots, \theta_{N-1}(\xi))^T$ and $A$ and $B$ are correspondingly two and three-diagonal constant $N\times N$ matrices (matrix indexes are running from 0 to $N-1$):
$$
A=-\frac{1}{2}
\begin{bmatrix}
0 & 0 &  &  &  &\\
3-\eta_1 & \eta_1+1 & 0  &  &  & &\\
 & 4-\eta_2 & \eta_2+2 & 0 &  & &\\
& & \ddots & & & &\\
& & 2+i-\eta_i & \eta_i +i & 0 & &\\
& & & \ddots & & & \\
& & & N-\eta_{N-2} & \eta_{N-2}+N-2 & 0 \\
& & &  & N+1-\eta_{N-1} & \eta_{N-1}+N-1
\end{bmatrix} ,
$$
\begin{equation}
B=-\frac{1}{2}
\begin{bmatrix}
-6 & 4 &  &  &  &\\
3-\eta_1 & 2\eta_1-6 & 1-\eta_1  &  &  & &\\
 & 4-\eta_2 & 2\eta_2 -6 & -\eta_2 &  & &\\
& & \ddots & & & &\\
& & 2+i-\eta_i & 2\eta_i -6 & 2-i-\eta_i & &\\
& & & \ddots & & & \\
& & & N-\eta_{N-2} & 2\eta_{N-2}-6 & 4-N-\eta_{N-2} \\
& & &  & N+1-\eta_{N-1} & 2\eta_{N-1}-6
\end{bmatrix} .
\label{B}
\end{equation}
Note that the two-diagonal matrix $A$ has a zero eigenvalue (zero Fuchsian exponent at $\xi=0$) which is a consequence of existence of an analytic at $w=0$ solution of (\ref{Fuchs}) (the one corresponding to solution of (\ref{Lrho},\ref{Leta}) we are looking for). This zero eigenvalue is non-degenerate, while other eigenvalues are negative, which confirms uniqueness of the above analytic solution (modulo constant factor). Indeed, this solution can be constructed starting from $\theta(0)$
\begin{equation}
\theta(\xi)=T(\xi) \theta(0), \quad T(\xi)=\sum_{n=0}^\infty T_n \xi^n, \quad A\theta(0)=0,
\label{transfer}
\end{equation}
where matrices $T_n$ are defined recurrently as
$$
T_0=1, \quad T_{n+1}=(A-n-1)^{-1}(A-B-n)T_n .
$$
Since all eigenvalues of $A-n-1$ are negative, the Taylor series (\ref{transfer}) exists and is unique (and converges absolutely for $\xi<1$).

The eigenvalues of the matrix $B$ are minus Fuchsian characteristic exponents at $\xi=1$. They are roots of the characteristic polynomial $P_N(\beta)=\det(\beta-B)$ which is determined by the three-term recurrence relation
\begin{equation}
\beta P_n = P_{n+1}+a_n P_{n-1}+b_n P_n , \quad P_0=1, \quad P_{-1}=0,
\label{Pn}
\end{equation}
where
\begin{equation}
a_1=3-\eta_1, \quad a_n=(\eta_n-n-2)(\eta_{n-1}+n-3)/4, \quad n>1 , \quad
b_n=3-\eta_n, \quad n\ge 0 .
\label{ab}
\end{equation}
When $a_n>0$, $0<n<N-1$, all eigenvalues of $B$ are real and non-degenerate and polynomials $P_n(\beta)$, $0\le n\le N$ are orthogonal wrt a non-signed measure (see e.g. \cite{Ch}). Although one can find a wide range of the Levy processes for which condition of positivity of $a_n$ holds, it is not always the case and $P_n$ are not necessarily orthogonal wrt a non-signed measure (i.e., $P_n(\beta)$ can be so-called ``formal" orthogonal polynomials).

\begin{lemma}

Matrix $B$, given by eq. (\ref{B}), has at least one real non-negative eigenvalue and $\beta(2)$ equals a non-negative eigenvalue of $B$.
\label{blowup}
\end{lemma}

\begin{remark}

\rm{Similar statement is valid for the bounded version of LLE for the matrix given by (\ref{Bext}). We will refer to both statements as Lemma \ref{blowup}}.

\end{remark}

\begin{proof}

It is known (see e.g. \cite{IY}) that an asymptotics of a general solution of the Fuchsian system (\ref{Fuchs}) in a neighborhood of singular point $\xi=1$ has the form
\begin{equation}
\theta\to\sum_{l=0}^{N-1} (1-\xi)^{-\beta_l} f^{(l)}(C,\xi), \quad \xi\to 1,
\label{fuchs_xi=1}
\end{equation}
where $\beta_l$ is an eigenvalue of $B$ and $C=\{C_0,\dots C_{N-1}\}$ are arbitrary constants. If matrix $B$ is diagonalizable and non-resonant \footnote{Matrix is resonant if it has eigenvalues which differ by a non-zero integer}, the vector functions $f^{(l)}$  are finite at $\xi=1$ and $f^{(l)}(C,\xi)=C_l\tilde{f}^{(l)}$, with $\tilde{f}^{(l)}$ being an eigenvector of $B$ corresponding to an eigenvalue $\beta_l$. If $B$ is not diagonalizable or/and resonant, $f^{(l)}(C,\xi)$ are polynomials in $x=\log(1-\xi)$ (i.e. they are of finite non-negative integer degrees in $x$).

From (\ref{fuchs_xi=1}) and (\ref{theta_1}) it follows that the corresponding solution of (\ref{Ltheta1}) has the following asymptotics at $\xi\to 1$ (or $r\to 1$ in polar coordinates $(r,\phi)$, $w=re^{i\phi}$)
$$
\Theta(r,\phi)\to\sum_{l=0}^{N-1}(1-r)^{-{\rm Re}\beta_l}\sum_{m=0}^{M_l}G_{lm}(\phi)x^m\cos\left(x{\rm Im}\beta_l-g_{lm}(\phi)\right), \quad x=\log(1-r^2), \quad r\to 1 ,
$$
where all functions $G_{lm}(\phi)$ are bounded. As $r\to 1$, the above finite sum will be dominated by term(s) with the maximal ${\rm Re}\beta_l$ and, if there are several of them, one has to choose the one(s) with the maximal $M_l$ among them. In other words
\begin{equation}
\Theta\to(1-r)^{-\max{\rm Re}\beta_l}x^M \mathcal{G}(x,\phi), \quad r\to 1 ,
\label{Theta_r=1}
\end{equation}
where $M$ is some non-negative integer and $\mathcal{G}(x,\phi)$ is a sum of finite number of the cosine terms
\begin{equation}
\mathcal{G}(x,\phi)=\sum_{\alpha_i\in\{{\rm Im}\beta_1,\dots{\rm Im}\beta_{N-1}\}}\mathcal{G}_{i}(\phi)\cos\left(\alpha_i x-\gamma_i(\phi)\right) .
\label{cosi}
\end{equation}
From (\ref{Theta_r=1}, \ref{cosi}) it follows that at least one eigenvalue with maximal real part has zero imaginary part. Indeed, let us suppose that the above is not true. Then all $\alpha_i$ in (\ref{cosi}) are non-zero. Since (\ref{cosi}) is the sum of a finite number of the cosine terms with $x$ running along the semi-infinite interval, $\mathcal{G}(x,\phi)$ is not of constant sign on this interval. Therefore, any real solution of (\ref{Ltheta1}) that truncates at $N$ in (\ref{theta_1}), including an analytic at $\xi=0$ solution we are looking for, oscillates, changing its sign as $\xi\to 1$. However, an analytic at $\xi=0$ solution cannot change sign at $\xi<1$, since both $\rho(w,\bar w)$ and $(1-w)(1-\bar w)$ are not negative at $\xi<1$. Therefore, there exists a real eigenvalue which is not less than the real part of any non-real eigenvalue.

Next, if all real $\beta_l$ were negative, then (according to (\ref{Theta_r=1})) any non-negative solution would have a negative blow-up rate and therefore $\beta(2)$ given by (\ref{beta}) would be negative. This is also a contradiction, since the integral means $\beta$-spectrum must be non-negative (see e.g. \cite{M}). Thus, matrix $B$ must have at least one non-negative eigenvalue, and $\beta(2)$ equals to one of non-negative eigenvalues of $B$.

\end{proof}

By consequence, in the $N=2$ case all eigenvalues are real\footnote{Also, one can check that in the $N=3$ case all roots of cubic characteristic polynomial are always real for any $\eta_1,\eta_2\ge0$, whatever are signs of $a_n$.} (as confirmed by Theorem \ref{N=2}). Note that here both eigenvalues can be positive, but $\beta(2)$ always equals the maximal one.

Another useful example where $\beta(2)$ and spectrum of $B$ can be found explicitly is the $N$-truncated unbounded SLE, i.e. the unbounded whole-plane SLE$_\kappa$ with $\kappa=2(N+2)/N^2$ (for details see \cite{LYSLE})
$$
\eta_n=\frac{N+2}{N^2}n^2, \quad n=0..N,
$$
\begin{equation}
\beta_l=\frac{N+2-(2N^2-3N-6)l+2(N+2)l^2}{N^2}, \quad l=0..N-1.
\label{betal}
\end{equation}
Here all eigenvalues $\beta_l$ are real, while not all $a_n$ are positive. The value of the integral means $\beta$ spectrum at $q=2$ equals the maximal eigenvalue $\beta_{N-1}$, i.e. $\beta(2)=\max\{\beta_l,l=0\dots N-1\}=(5N-2)/N$ (see \cite{LYSLE,DNNZ}).

The above example shows that spectrum of $B$ can be degenerate and resonant. For instance, the spectrum is resonant for
$N=2,3,10,12,14,15,\dots$.  \\ 

To present an example of $B$ with complex eigenvalues, first consider the case when $N=6$ in (\ref{betal}), i.e. for $\eta_n=4n^2/18$, $n\le6$. Here the spectrum is degenerate and $\beta_0=\beta_3=2/9$, $\beta_1=\beta_2=-2/3$ (two other eigenvalues $\beta_4=2$ and $\beta_5=14/3$ are non-degenerate). Let us now introduce the following perturbation of the above LLE
$$
\eta_n=\left(\frac{4}{9}-\delta\kappa\right)\frac{n^2}{2}+18\delta\kappa, \quad \delta\kappa>0 \quad 0<n\le 6,
$$
which can be, for instance, the LLE driven by a combination of the Brownian motion with $\kappa=4/9-\delta\kappa$ and a compound Poisson process with jumps uniformly distributed over the circle. The Levy measure of the latter equals $d\eta(\phi)=(9\delta\kappa/\pi)d\phi$ (see eq.(\ref{LK})). When $\delta\kappa$ are small, the spectrum of $B$ has two positive and four imaginary eigenvalues. The latter four, up to $\mathcal{O}(\delta\kappa)$, equal
$$
\frac{2}{9}\pm \frac{21}{128}\sqrt{-30\delta\kappa}, \quad -\frac{2}{3}\pm \frac{5}{128}\sqrt{-70\delta\kappa}, \quad \delta\kappa>0 .
$$ \\

Returning to the generic LLE, we note that since analytic at $\xi=0$ solution of (\ref{Fuchs}) is a linear combination of solutions corresponding to different Fuchsian exponents at $\xi=1$, one may assume that the $\beta$-spectrum corresponds to the lowest real exponent at $\xi=1$ (i.e. maximal real eigenvalue of $B$).  Indeed, an absence of the corresponding component in the above linear combination could happen only by an ``accident".  Although, as we have seen above, there is no such an ``accident" in the $N=2$ case as well in the case of the $N$-truncated SLE, its absence for the generic $N>2$ case (or generic $N>4$ case for the bounded LLE, see next section) requires a proof that does not seem to be elementary.

Next, we will consider a family of LLEs which are small deformations of the $N$-truncated SLEs. It is more convenient to consider deformation of bounded version where, in contrast to (\ref{betal}), all eigenvalues but one are negative. The latter fact facilitates the study.

\end{section}

\begin{section}{Bounded LLE}

For $q=2$ one looks for a solution of (\ref{Lrho},\ref{Leta}) that is analytic at $w\to\infty$, representing $\rho(w,\bar w)$ in the following form
$$
\rho(w,\bar w)=(1-w^{-1})(1-\bar w^{-1})\Theta(w,\bar w).
$$
Then for $q=2$ we get
\begin{equation}
-\hat\eta\left[\left(1-\frac{1}{w}\right)\left(1-\frac{1}{\bar w}\right)\Theta\right]+\left(1+\frac{1}{w}\right)\left(1-\frac{1}{\bar w}\right)w\frac{\partial\Theta}{\partial w}+\left(1+\frac{1}{\bar w}\right)\left(1-\frac{1}{w}\right)\bar
w\frac{\partial\Theta}{\partial \bar w}+\left(\frac{1}{w}+\frac{1}{\bar w}-\frac{2}{w \bar w}\right)\Theta=0 .
\label{Ltheta_ext}
\end{equation}
For $\Theta$ given by the Fourier series
$$
\Theta=\theta_0(\xi)+\sum_{i=1}^{\infty}\left(w^{-i}+\xi^{-i} w^{i} ,
\right) \theta_i(\xi), \quad \xi=r^2=w\bar w
$$
we get the following recurrence relation for $\theta_i(\xi)$
$$
2\xi(\xi-1)\theta'_i(\xi)-\left(\eta_i+2-i+(\eta_i+i)\xi\right)\theta_i(\xi)+(\eta_i+i+2)\theta_{i+1}(\xi)+\xi(\eta_i-i+2)\theta_{i-1}(\xi)=0, \quad \theta_{-i}(\xi)=\xi^{-i}\theta_i(\xi) .
$$
This relation truncates if $\eta_N=N-2$. In case of the truncation it can be rewritten in the Fuchsian form (\ref{Fuchs}) with
$$
A=-\frac{1}{2}
\begin{bmatrix}
2 & -4 &  &  &  &\\
0 & \eta_1+1 & -3-\eta_1  &  &  & &\\
 & 0 & \eta_2 & -4-\eta_2 &  & &\\
& & \ddots & & & &\\
& & 0 & \eta_i + 2 - i & -2-i-\eta_i & &\\
& & & \ddots & & & \\
& & & 0 & \eta_{N-2}-N+4 & -N-\eta_{N-2} \\
& & &  & 0 & \eta_{N-1}+3-N
\end{bmatrix} ,
$$
\begin{equation}
B=-\frac{1}{2}
\begin{bmatrix}
2 & -4 &  &  &  &\\
-1-\eta_1 & 2\eta_1+2 & -3-\eta_1  &  &  & &\\
 & -\eta_2 & 2\eta_2+2 & -4 -\eta_2 &  & &\\
& & \ddots & & & &\\
& & i-2-\eta_i & 2\eta_i +2 & -2-i-\eta_i & &\\
& & & \ddots & & & \\
& & & N-4-\eta_{N-2} & 2\eta_{N-2}+2 & -N-\eta_{N-2} \\
& & &  & N-3-\eta_{N-1} & 2\eta_{N-1}+2
\end{bmatrix} .
\label{Bext}
\end{equation}
Now matrix $A-B$ (i.e. residue matrix at $\xi=\infty$) has a single zero eigenvalue and $N-1$ positive eigenvalues which confirms that an analytic at $w=\infty$ solution is an $N$-truncated solution (demonstration of this fact repeats arguments used for unbounded LLE in the previous section).

Note that, since $\eta_n>0$ when $n>0$, while $\eta_N=N-2$, the first non-trivial case of truncation takes place at $N=3$ (i.e. for $N>2$). This differs from the unbounded case considered in the previous section, where non-trivial truncations take place for all $N>0$.

The characteristic polynomial $P_N(\beta)$ can be found with the help
of the three-term recurrence relation (\ref{Pn}) with the following coefficients
\begin{equation}
a_1=\eta_1+1, \quad a_n=(\eta_n-n+2)(\eta_{n-1}+n+1)/4, \quad n>1 , \quad b_n=-\eta_n-1, \quad n\ge 0 .
\label{ab_ext}
\end{equation}

Similarly to the unbounded version an analog of Lemma \ref{blowup} for matrix $B$, given by (\ref{Bext}) holds. Let us now prove the following

\begin{theorem}
\label{bounded_q2}
For the bounded whole-plane LLE with
$$
\eta_N=N-2
$$
and
$$
\eta_n=\frac{\kappa_N n^2}{2}+\delta\eta_n, \quad \kappa_N=2\frac{N-2}{N^2}, \quad 0<n<N,
$$
where $\delta\eta_n\not=0$ are sufficiently small \footnote{One can easily check with the help of the Levy-Khinchine formula (\ref{LK}) that such families of Levy processes exist, e.g. by setting $\kappa=\kappa_N-\delta \kappa$, $\delta\kappa>0$ and $d\eta(\phi)>0$ a such that $\int_0^{2\pi}(1-\cos N\phi)d\eta(\phi)=\delta \kappa N^2/2$.}, the value of the integral means $\beta$-spectrum at $q=2$ is given by the greatest eigenvalue of the matrix (\ref{Bext}).
\end{theorem}

\begin{proof}

In the case of the bounded SLE$_k$ at $q=2$ and $\kappa=\kappa_N=2(N-2)/N^2$, the spectrum of the matrix $B$ reads as \cite{LYSLE}
\begin{equation}
\beta_l=\frac{2\kappa_N+(3\kappa_N-4)l+\kappa_Nl^2}{4}, \quad l=0..N-1
\label{betaextl}
\end{equation}
with maximal eigenvalue $\beta(2)=\beta_0$ being the only non-negative eigenvalue (for derivation of the spectra see \cite{LYSLE}, \cite{BS}).

Since the spectrum is real and non-degenerate, it remains such when we add small perturbations $\delta\eta_n$
to $\eta_n=\kappa_N n^2/2$, $n<N$. Indeed, let $P_N(\beta)$ be a characteristic polynomial of matrix $B$
$$
P_N=\beta^N+\sum_{n=0}^{N-1} p_n \beta^n, \quad p_n=p_n(\eta_1, \dots \eta_{N-1}) ,
$$
then
$$
\frac{\partial \beta_l}{\partial \eta_j}=-\frac{\sum_{n=0}^{N-1} \frac{\partial p_n}{\partial \eta_j}\beta_l^n}{\prod_{k\not=l}(\beta_l-\beta_k)} .
$$
Since $\beta_i\not=\beta_j$ when $i\not=j$, all derivatives of the spectrum wrt $\eta_n$ are finite and real and therefore spectrum remains real and non-degenerate for small perturbations $\eta_n \to \eta_n+\delta \eta_n$, $\delta\eta_n\not=0$, $0<n<N$.

In eq. (\ref{betaextl}) the maximal eigenvalue $\beta_0$ is positive while the rest of eigenvalues are negative. Since $\partial\beta_l/\partial \eta_i$ are finite, the sign of eigenvalues does not change for small deformations $\eta_n\to\eta_n+\delta \eta_n$, so the maximal eigenvalue still remains the only non-negative eigenvalue. Thus, from Lemma \ref{blowup} it follows that $\beta(2)$ equals the maximal eigenvalue, which completes the proof.

\end{proof}

The next proposition considers generic nontrivial $N<5$ cases

\begin{theorem}
\label{N3_N4}

For the bounded whole-plane LLE with $\eta_N=N-2$, where $N<5$, the value of the integral means $\beta$-spectrum at $q=2$ is the maximal real eigenvalue of the three-diagonal matrix (\ref{Bext}).

\end{theorem}

\begin{proof}

By direct computation, from (\ref{Pn},\ref{ab_ext}) we get
$$
P_3={\beta}^{3}+\left (\eta_{{1}}+\eta_{{2}}+3\right ){\beta}^{2}+\left (2
+\eta_{{1}}+\frac{5}{4}\,\eta_{{2}}+\frac{3}{4}\,\eta_{{2}}\eta_{{1}}\right )\beta-\frac{1}{4}
\,\eta_{{2}}\eta_{{1}}-\frac{3}{4}\,\eta_{{2}},
$$
$$
P_4={\beta}^{4}+\left (\eta_{{2}}+\eta_{{1}}+\eta_{{3}}+4\right ){\beta}^{
3}+\left (\frac{3}{4}\,\eta_{{2}}\eta_{{1}}+\frac{5}{2}\,\eta_{{2}}+\eta_{{3}}\eta_{{1
}}+\frac{3}{4}\,\eta_{{3}}\eta_{{2}}+2\,\eta_{{3}}+6+2\,\eta_{{1}}\right ){
\beta}^{2}+
$$
$$
+\left (\frac{1}{2}\,\eta_{{3}}\eta_{{2}}\eta_{{1}}
+\frac{3}{4}\,\eta_{{3}}
\eta_{{2}}+\eta_{{2}}+4+\frac{3}{4}\,\eta_{{2}}\eta_{{1}}+2\,\eta_{{1}}\right
)\beta-\frac{3}{4}\,\eta_{{3}}\eta_{{2}}-\frac{1}{4}\,\eta_{{3}}\eta_{{2}}\eta_{{1}}-\frac{3}
{4}\,\eta_{{2}}-\frac{1}{4}\,\eta_{{2}}\eta_{{1}} .
$$
Since $\eta_n>0$, $n>0$, the polynomial $P_4$ has a single negative coefficient. By Descartes' rule of sign it has only one positive root and the rest of roots are negative. The same \footnote{Actually, there is no need to consider $P_3$ separately since existence of only one non-negative eigenvalue for $N$ automatically implies the same for $N'<N$.} applies to $P_3$. By Lemma \ref{blowup}, $\beta(2)$ equals the maximal root of characteristic polynomial.

\end{proof}

\begin{figure}
\centering
\includegraphics[width=115mm]{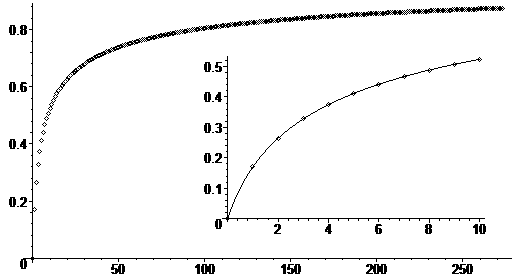}
 \caption{\small Dependence of $\beta(2)$ of the PLE$_\lambda$ on the arrival rate $\lambda=N-2$ of the Poisson process. Numerical results suggest that $\beta(2)\to 1$ as $\lambda\to\infty$.
 Results of numerical computations of the maximal eigenvalue in general case (i.e. case of infinite matrices) is shown separately in the zoomed region by the solid line.}
 \label{PLE}
 \end{figure}

Concluding this section we would like to mention another set of examples which are similar to the Hastings-Levitov (HL) models \cite{HL}, namely, LLE driven by a compound Poisson process with jumps distributed uniformly over the circle (we call them PLE$_\lambda$, where $\lambda$ is the arrival rate of the Poisson process) \footnote{Similarly to the HL model, the PLE$_\lambda$ is a composition of random elementary ``spike" mappings uniformly distributed over the circle. Note, however, that the HL models are not LLEs. LLEs driven by compound Poisson processes have been considered in \cite{JS}.}. Here, $\eta_n=\lambda=N-2$, $n>0$ and the recurrence coefficients $a_n$, $n=1\dots N-1$ are positive. Therefore all roots of $P_N(\beta)$ are real and simple. We have verified for $N$ up to several hundred that a polynomial $P_N(\beta)$ has the only one change of sign in the sequence of its coefficients. By consequence, in these (verified) models, $\beta(2)$ equals maximal eigenvalue which is the only non-negative eigenvalue (see Figure \ref{PLE}).

\end{section}

\begin{section}{Conclusions}
\label{conj}

In the present paper we considered a wide class of the whole-plane LLEs driven by a generic Levy process on a circle, without the drift, restricted by a condition that there exist $N$ for which $\eta_N=N+2$ in the unbounded version of LLE and $\eta_N=N-2$ in the bounded version. We showed that the second moment of derivative of the stochastic conformal mapping for the whole-plane LLE is expressed in terms of solution of $N$-dimensional Fuchsian system with three singular points. The value of the integral means $\beta$ spectrum at $q=2$ (i.e. $\beta(2)$) is a non-negative eigenvalue of a three-diagonal matrix.
We also were able to show that $\beta(2)$ is actually the maximal eigenvalue for several wide classes of processes.

Therefore, one might try to generalize the Theorem \ref{N=2}, \ref{bounded_q2} and \ref{N3_N4} by dropping the condition of smallness of $\delta\eta_n$ in the Theorem \ref{bounded_q2}. We recall that in the case of Theorems \ref{N=2}, \ref{N3_N4} condition of smallness is absent, so we may put forward the following

\begin{conjecture}

Let the whole-plane LLE is driven by a Levy process without drift for which there exists an integer $N$ a such that

\begin{itemize}

\item $\eta_N=N+2$ in the unbounded version of LLE

\item $\eta_N=N-2$ in the bounded version of LLE.

\end{itemize}

Then, the value of the integral means $\beta$-spectrum at $q=2$ equals the maximal real eigenvalue of the three-diagonal matrix $B$ given by (\ref{B}) or (\ref{Bext}) correspondingly.

\end{conjecture}

Moreover, one can generalize the above conjecture for the case when the truncation conditions do not necessarily hold. Indeed, in this case matrix $B$ is infinite and we can look for a maximal eigenvalue of this infinite matrix.
In more detail, one has to take an infinite set of ``formally orthogonal" polynomials and look at the sequence of maximal roots (or equally at the sequence of maximal eigenvalues $\beta_{\rm max}(N)$ of the sub-matrices of dimension $N$ of an infinite matrix $B$). One expects this sequence to converge to $\beta(2)$ as $N\to\infty$.

We have checked this hypothesis numerically for SLE$_\kappa$ (where $\beta(q)$-spectrum is known for any $\kappa$ and $q$) for different $\kappa$. Also, for PLE$_\lambda$ numerical computations suggest that the maximal eigenvalue is a smooth function of $\lambda$ (see Figure \ref{PLE}). When the $N$-truncation takes place, the maximal eigenvalue of the infinite matrix equals the one of its sub-matrix of size $N$, i.e. the sequence $\beta_{\rm max}(M)$ reaches its limit at $M=N$.

Therefore, one may state the following

\begin{conjecture}

Let the whole-plane LLE is driven by a Levy process without drift. Then $\beta(2)$ is a limit of the sequence of maximal real eigenvalues of matrices $B_N$, where $B_N$ is given by (\ref{B}) for the unbounded version of LLE or (\ref{Bext}) for the bounded version  (or, equivalently, the sequence of maximal roots of polynomials set by the recurrence relation (\ref{Pn},\ref{ab}) or (\ref{Pn},\ref{ab_ext}) correspondingly).

\end{conjecture}

\end{section}


\begin{thebibliography}{101}

\bibitem{A} Applebaum D, {\it Levy Processes - From Probability to Finance and Quantum Groups}, Notices of the AMS,  51 (11): 1336–1347 (2004)

\bibitem{BS}
D. Beliaev, S. Smirnov, \emph{Harmonic measure and SLE}, Commun.
Math. Phys. 290, 577–595 (2009).

\bibitem{BDZ} D. Beliaev, B. Duplantier, and M. Zinsmeister, {\it Integral means spectrum of whole-plane SLE}, Comm. Math. Phys. Volume 353, Issue 1, pages 119–133, (2017).

\bibitem{C} John Cardy, \emph{SLE for theoretical physicists}, Ann.Phys. 318 (2005) 81-118

\bibitem{Ch} T. Chihara, {\it An Introduction to Orthogonal Polynomials}, Gordon and Breach, NY, 1978

\bibitem{D} B.Duplantier, \emph{Conformally invariant fractals and potential theory}, Phys.Rev.Lett., 84(7): 1363-1367,(2000)

\bibitem{DNNZ} Bertrand Duplantier, Thi Phuong Chi Nguyen, Thi Thuy Nga Nguyen, Michel Zinsmeister, \emph{The Coefficient Problem and Multifractality of Whole-Plane SLE and LLE}, Annales Henri Poincare,  Volume 16, Issue 6, pp 1311 - 1395, (2015)

\bibitem{G}
Ilya A.~Gruzberg, \emph{Stochastic geometry of critical curves,
Schramm--Loewner evolutions and conformal field theory}, J.~Phys.~A:
Math.~Gen. \textbf{39}, no.~41 (2006) 12601--12655.

\bibitem{HHD} Halsey, T. C., Honda, K., Duplantier, B. (1996), {\it Multifractal dimensions for branched growth}, J.Stat.Phys., 85(5-6), pp 681–743

\bibitem{HJKPS} T.C. Halsey, M.H. Jensen, L.P. Kadanoff, I. Procaccia, and B.I. Shraiman, \emph{Fractal measures and their singularities: the characterization of strange sets}, Phys. Rev. A (3), 33(2): p. 1141, 1986.

\bibitem{Has}
Matthew B.~Hastings, \emph{Exact Multifractal Spectra for Arbitrary
Laplacian Random Walks}, \emph{Phys. Rev. Lett.} \textbf{88} (2002)

\bibitem{HL} M.B. Hastings, L.S. Levitov, {\it Laplacian growth as one-dimensional turbulence}, Physica D, 116 (1998), 244-252

\bibitem{IY} Yu. Ilyashenko, S. Yakovenko, {\it Lectures on analytic differential equations}, Graduate Studies in Mathematics, 86. American Mathematical Society, Providence, RI, 2008. xiv+625 pp. ISBN: 978-0-8218-3667-5

\bibitem{JS} Fredrik Johansson, Alan Sola, {\it Rescaled Levy-Loewner hulls and random growth}, Bulletin des Sciences Mathématiques, Volume 133, Issue 3, April 2009, Pages 238-256


\bibitem{GL}
Gregory F.~Lawler, \emph{Conformally invariant processes in the
plane}, Mathematical Surveys and Monographs, \textbf{114}, Amer.
Math. Soc., Providence, RI, (2005).

\bibitem{GL2}
Gregory F.~Lawler, \emph{Conformal invariance and 2D statistical
physics}, \emph{Bull. Amer. Math. Soc.} \textbf{46} (2009) 35--54.

\bibitem{LSLE}
Igor Loutsenko, \emph{SLE$_\kappa$: correlation functions in  the
coefficient problem}, J. Phys. A: Math. Theor. 45 275001,
2012,(2012)

\bibitem{LYSLE} Igor Loutsenko, Oksana Yermolayeva,
\emph{On Exact Multi-Fractal Spectrum of the Whole-Plane SLE}, J.Stat. Mech., (2013) doi: 10.1088/1742-5468/2013/04/P4007

\bibitem{LYLLE} Igor Loutsenko, Oksana Yermolayeva, {\it New exact results in spectra of stochastic Loewner evolution}, J. Phys. A, 47(16):165202, 15, 2014.

\bibitem{M} N. G. Makarov, {\it Fine structure of harmonic measure}, St. Petersburg Math.
J., 10(2):217–268, 1999.


\bibitem{ORGK}
P. Oikonomou, I. Rushkin, I. A. Gruzberg, and L. P. Kadanoff,
{\it Global properties of stochastic Loewner evolution driven
by L\'evy processes}, J. Stat.
Mech. (2008) P01019

\bibitem{ROKG}
I. Rushkin, P. Oikonomou, L. P. Kadanoff, and I. A. Gruzberg, {\it
Stochastic Loewner evolution driven by Levy processes}, J. Stat.
Mech. (2006) P01001

\bibitem{Wiki} \url{http://en.wikipedia.org/wiki/Multifractal_system}

\end{thebibliography}
\end{document}